%% file: main.tex
\documentclass[sigconf]{acmart}

\usepackage{booktabs}
\usepackage{amssymb}
\usepackage{amsmath} 
\usepackage{courier}
\usepackage{enumitem}
\usepackage{bbding}
\usepackage{pifont}
\usepackage{setspace}
\usepackage{array,multirow,multicol}
\usepackage{caption, subcaption}
\usepackage{algorithm} 
\usepackage{algpseudocode} 
\usepackage{lipsum}

\newcommand{\PreserveBackslash}[1]{\let\temp=\\#1\let\\=\temp}
\newcolumntype{C}[1]{>{\PreserveBackslash\centering}p{#1}}
\newcolumntype{R}[1]{>{\PreserveBackslash\raggedleft}p{#1}}
\newcolumntype{L}[1]{>{\PreserveBackslash\raggedright}p{#1}}
\newcommand{\tabincell}[2]{\begin{tabular}{@{}#1@{}}#2\end{tabular}}  
\newtheorem{definition}{Definition}
\newtheorem{theorem}{Theorem}

\AtBeginDocument{%
  \providecommand\BibTeX{{%
    \normalfont B\kern-0.5em{\scshape i\kern-0.25em b}\kern-0.8em\TeX}}}

\copyrightyear{2022} 
\acmYear{2022} 
\setcopyright{rightsretained} 

\acmConference[KDD '22]{Proceedings of the 28th ACM SIGKDD Conference on Knowledge Discovery and Data Mining}{August 14--18, 2022}{Washington, DC, USA}
\acmBooktitle{Proceedings of the 28th ACM SIGKDD Conference on Knowledge Discovery and Data Mining (KDD '22), August 14--18, 2022, Washington, DC, USA}
\acmISBN{978-1-4503-9385-0/22/08}
\acmDOI{10.1145/3534678.3539253}

\usepackage{etoolbox}
\makeatletter
\patchcmd{\maketitle}{\@copyrightpermission}{
   \begin{minipage}{0.3\columnwidth}
     \href{https://creativecommons.org/licenses/by/4.0/}{\includegraphics[width=0.90\textwidth]{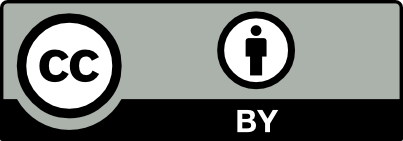}}
   \end{minipage}\hfill
   \begin{minipage}{0.7\columnwidth}
     \href{https://creativecommons.org/licenses/by/4.0/}{This work is licensed under a Creative Commons Attribution International 4.0 License.}
   \end{minipage}
  
   \vspace{5pt}
}{}{}

\makeatother

\begin{document}

\title{Towards Representation Alignment and Uniformity in Collaborative Filtering}


\author{Chenyang Wang}
\affiliation{%
  \institution{DCST, BNRist, Tsinghua University}
  \city{Beijing 100084}
  \country{China}}
\email{wangcy18@mails.tsinghua.edu.cn}

\author{Yuanqing Yu}
\affiliation{%
  \institution{DCST, BNRist, Tsinghua University}
  \city{Beijing 100084}
  \country{China}}
\email{yuyq18@mails.tsinghua.edu.cn}

\author{Weizhi Ma}
\affiliation{%
  \institution{AIR, Tsinghua University}
  \city{Beijing 100084}
  \country{China}}
\email{mawz@tsinghua.edu.cn}

\author{Min Zhang*}
\affiliation{%
  \institution{DCST, BNRist, Tsinghua University}
  \city{Beijing 100084}
  \country{China}}
\email{z-m@tsinghua.edu.cn}

\author{Chong Chen}
\affiliation{%
  \institution{DCST, BNRist, Tsinghua University}
  \city{Beijing 100084}
  \country{China}}
\email{cc17@mails.tsinghua.edu.cn}

\author{Yiqun Liu}
\affiliation{%
  \institution{DCST, BNRist, Tsinghua University}
  \city{Beijing 100084}
  \country{China}}
\email{yiqunliu@tsinghua.edu.cn}

\author{Shaoping Ma}
\affiliation{%
  \institution{DCST, BNRist, Tsinghua University}
  \city{Beijing 100084}
  \country{China}}
\email{msp@tsinghua.edu.cn}

\def\authors{Chenyang Wang, Yuanqing Yu, Weizhi Ma, Min Zhang, Chong Chen, Yiqun Liu, Shaoping Ma}


\renewcommand{\shortauthors}{Chenyang Wang et al.}

\begin{abstract}
Collaborative filtering (CF) plays a critical role in the development of recommender systems.
Most CF methods utilize an encoder to embed users and items into the same representation space, and the Bayesian personalized ranking (BPR) loss is usually adopted as the objective function to learn informative encoders.
Existing studies mainly focus on designing more powerful encoders (e.g., graph neural network) to learn better representations.
However, few efforts have been devoted to investigating the desired properties of representations in CF, which is important to understand the rationale of existing CF methods and design new learning objectives.
In this paper, we measure the representation quality in CF from the perspective of \textit{alignment} and \textit{uniformity} on the hypersphere.
We first theoretically reveal the connection between the BPR loss and these two properties.
Then, we empirically analyze the learning dynamics of typical CF methods in terms of quantified alignment and uniformity, which shows that better alignment or uniformity both contribute to higher recommendation performance.
Based on the analyses results, a learning objective that directly optimizes these two properties is proposed, named DirectAU.
We conduct extensive experiments on three public datasets, and the proposed learning framework with a simple matrix factorization model leads to significant performance improvements compared to state-of-the-art CF methods.
Our implementations are publicly available\footnote{https://github.com/THUwangcy/DirectAU}.
\end{abstract}



\maketitle

\newcommand\blfootnote[1]{%
    \begingroup
    \renewcommand\thefootnote{}\footnote{#1}%
    \addtocounter{footnote}{-1}%
    \endgroup
}

\blfootnote{*Corresponding author.}

\vspace{-5mm}
\input{sections/S1-Introduction}
\input{sections/S2-Preliminaries}
\input{sections/S3-Analyses}
\input{sections/S4-Methodology}
\input{sections/S5-Experiments}
\input{sections/S6-RelatedWork}

\input{sections/S7-Conclusion}

\begin{acks}
This work is supported by the Natural Science Foundation of China (Grant No. U21B2026, 62002191) and Tsinghua University Guoqiang Research Institute.
We would like to thank the VMWare gift funding's partly support to the authors.
\end{acks}

\bibliographystyle{ACM-Reference-Format}
\bibliography{bibliography}

\appendix
\balance
\input{sections/Z-Appendix}

\end{document}

%% file: sections/S1-Introduction.tex
\section{Introduction}
Recommender system has become an essential part of users' engagements with web services, such as product recommendation~\cite{mcauley2015image}, video recommendation~\cite{covington2016deep}, and so on.
To help users discover potential items of interests, collaborative filtering (CF) is widely adopted in personalized recommendation~\cite{schafer2007collaborative}.
The core idea of CF is that similar users tend to have similar preferences.
Compared to content-based recommendation methods, CF only relies on past user behaviors to predict users' preferences on candidate items.
The simplicity and effectiveness of CF make it a canonical technique in recommender systems~\cite{su2009survey}.

Most CF methods utilize an encoder to embed users and items to a shared space and then optimize an objective function to learn informative user and item representations~\cite{mao2021simplex}.
The simplest encoder can be an embedding table that directly maps user and item IDs to embeddings~\cite{koren2009matrix}, and Bayesian personalized ranking (BPR)~\cite{rendle2009bpr} is usually adopted as the objective function to discriminate between positive interactions and unobserved ones.
Existing studies about CF mainly focus on designing more powerful encoders to model complex collaborative signals between users and items.
Specifically, neural-based interaction encoders emerge in recent years, such as multi-layer perceptron (MLP)~\cite{he2017neural}, attention mechanism~\cite{chen2017attentive}, graph neural network (GNN)~\cite{wang2019neural, he2020lightgcn}, and so on.
Meanwhile, some recent works point out that the nowadays complex encoders in CF actually lead to marginal performance improvements~\cite{mao2021simplex}.
As a result, researchers also begin to investigate other objective functions beyond the common pairwise BPR loss (e.g., InfoNCE loss~\cite{zhou2021contrastive}, cosine contrastive loss~\cite{mao2021simplex}), which have been shown to bring more robust improvements than complex encoders.

However, few research efforts have been devoted to investigating the desired properties of user and item representations derived by the encoder.
This is important to justify the rationale behind existing CF methods and design new learning objectives that favor these properties.
Intuitively, representations of positive-related user-item pairs should be close to each other, and each representation should preserve as much information about the user/item itself as possible.
Assuming all the representations are $l_2$ normalized, these two properties can be referred to as 1) \textit{alignment} and 2) \textit{uniformity} on the unit hypersphere~\cite{wang2020understanding}.
To learn informative user and item representations, both alignment and uniformity are of great importance.
If only alignment is considered, perfectly aligned encoders are easy to be achieved by mapping all the users and items to the same embedding.
The goal of existing loss functions in CF can be seen to avoid such trivial constants (i.e., preserving uniformity) while optimizing for better alignment.
In practice, negative samples are usually utilized to achieve this goal.
For example, the BPR loss~\cite{rendle2009bpr} pairs each positive interaction with a randomly sampled negative item, and the predicted score of the interacted item is encouraged to be higher than the negative one.

In this work, we analyze the \textit{alignment} and \textit{uniformity} properties in CF inspired by recent progress in contrastive representation learning~\cite{wang2020understanding, gao2021simcse}.
We first theoretically show that the BPR loss actually favors these two properties, and perfectly aligned and uniform encoders form the exact minimizers of the BPR loss.
Then, we empirically analyze the learning dynamics of typical CF methods in terms of alignment and uniformity via corresponding quantifying metrics proposed in~\cite{wang2020understanding}.
We find different CF methods demonstrate distinct learning trajectories, and either better alignment or better uniformity benefits the representation quality.
For instance, the simplest BPR quickly converges to promising alignment and mainly improves uniformity afterwards.
Other advanced methods achieve better alignment or uniformity via various techniques, such as hard negative samples and graph-based encoders, which lead to better performance accordingly.
Based on the analyses results, we propose a learning objective that directly optimizes these two properties, named DirectAU.
Extensive experiments are conducted on three public real-world datasets.
Experimental results show that a simple matrix factorization based encoder (i.e., embedding table) that optimizes the proposed DirectAU loss yields remarkable improvements (up to 14\%) compared to state-of-the-art CF methods.

The main contributions of this work can be summarized as follows:
\begin{itemize}
    \item We theoretically show that perfectly aligned and uniform encoders form the exact minimizers of the BPR loss. We also empirically analyze the learning dynamics of typical CF methods in terms of quantified alignment and uniformity.
    \item Based on the analyses results, a simple but effective learning objective that directly optimizes these two properties is proposed, named DirectAU.
    \item Extensive experiments on three public datasets show that the proposed DirectAU well balances between alignment and uniformity. When optimizing the DirectAU objective, even the simplest matrix factorization based encoder leads to significant performance improvements compared to state-of-the-art CF methods.
\end{itemize}

%% file: sections/S2-Preliminaries.tex
\section{Preliminaries}
In this section, we first formulate the collaborative filtering problem.
Then we introduce how to measure alignment and uniformity based on recent progress in self-supervised learning~\cite{wang2020understanding}.

\subsection{Collaborative Filtering}
Let $\mathcal{U}$ and $\mathcal{I}$ denote the user and item set, respectively.
Given a set of observed user-item interactions $\mathcal{R}=\{(u,i)~|~u~{\rm interacted~with}~i\}$, CF methods aim to infer the score $s(u,i)\in\mathbb{R}$ for each unobserved user-item pair indicating how likely the user $u$ tends to interact with the item $i$.
Then, items with the highest scores for each user will be recommended based on the predictions.

In general, most CF methods use an encoder network $f(\cdot)$ that maps each user and item into a low-dimensional representation $f(u),f(i)\in\mathbb{R}^d$ ($d$ is the dimension of the latent space).
For example, the encoder in matrix factorization models is usually an embedding table, which directly maps each user and item to a latent vector based on their IDs.
The encoder in graph-based models further utilizes the neighborhood information.
Then, the predicted score is defined as the similarity between the user and item representation (e.g., dot product, $s(u,i)=f(u)^Tf(i)$).
As for the learning objective, most studies adopt the pairwise BPR~\cite{rendle2009bpr} loss to train the model:
\begin{equation}
    \mathcal{L}_{BPR} = \dfrac{1}{|\mathcal{R}|}\sum_{(u,i)\in\mathcal{R}}-\log\left[{\rm sigmoid}\left(s(u,i) - s(u,i^-)\right)\right],
    \label{equ:bpr}
\end{equation}
where $i^-$ is a randomly sampled negative item that the user has not interacted with.
This loss function aims to optimize the probability that the target item gets a higher score than random negative items.

\subsection{Alignment and Uniformity}
Recent studies~\cite{wang2020understanding, gao2021simcse} in unsupervised contrastive representation learning identify that the quality of representations is highly related to two key properties, i.e., alignment and uniformity.
Given the distribution of data $p_{\rm data}(\cdot)$ and the distribution of positive pairs $p_{\rm pos}(\cdot,\cdot)$, alignment is straightforwardly defined as the expected distance between normalized embeddings of positive pairs:
\begin{equation}
    l_{\rm align} \triangleq \mathop{\mathbb{E}}_{(x, x^+)\sim p_{\rm pos}}||\Tilde{f(x)} - \Tilde{f(x^+)}||^2,
\end{equation}
where $\tilde{f(\cdot)}$ indicates $l_2$ normalized representations.
On the other hand, the uniformity loss is defined as the logarithm of the average pairwise Gaussian potential:
\begin{equation}
    l_{\rm uniform} \triangleq \log\mathop{\mathbb{E}}_{x,y\sim p_{\rm data}}e^{-2||\tilde{f(x)} - \tilde{f(y)}||^2}.
\end{equation}
These two metrics are well aligned with the objective of representation learning: positive instances should be close to each other while random instances should scatter on the hypersphere.
In this work, we will connect the BPR loss with these two metrics and use them to analyze the learning dynamics of typical CF methods.


%% file: sections/S3-Analyses.tex
\section{Alignment and Uniformity in Collaborative Filtering}
In this section, we first theoretically show that the BPR loss favors representation alignment and uniformity on the hypersphere.
Then, we empirically observe how these two properties evolve during training for different CF methods.

\subsection{Theoretical Analyses}
Assuming the distribution of positive user-item pairs is $p_{\rm pos}$, and the distribution of users and items is denoted as $p_{\rm user}$ and $p_{\rm item}$ respectively, we first define the notion of optimality for alignment and uniformity in CF as follows:

\begin{definition}[Perfect Alignment]
An encoder $f$ is perfectly aligned if $\tilde{f(u)}=\tilde{f(i)}$ a.s. over $(u,i)\sim p_{\rm pos}$.
\end{definition}

\begin{definition}[Perfect Uniformity]
An encoder $f$ is perfectly uniform if the distribution of $\tilde{f(u)}$ for $u\sim p_{\rm user}$ and the distribution of $\tilde{f(i)}$ for $i\sim p_{\rm item}$ are the uniform distribution $\sigma_{d-1}$ on $\mathcal{S}^{d-1}$.
\end{definition}

Here $\mathcal{S}^{d-1} = \{x\in\mathbb{R}^d: ||x||=1\}$ is the surface of the $d$-dimensional unit ball.
Note that perfectly aligned encoders can be easily achieved by mapping all the inputs to the same representation, at the cost of the worst uniformity.
Perfectly uniform encoders can also be achieved considering the number of users/items is usually large and $d$ is small in real-world applications. 
The following theorem shows that the BPR loss favors these two properties if perfect alignment and uniformity are realizable.

\begin{theorem}
If perfectly aligned and uniform encoders exist, they form the exact minimizers of the BPR loss $\mathcal{L}_{BPR}$.
\end{theorem}

\begin{proof}
Assuming the similarity function $s(u,i)$ is cosine similarity (user/item representations are normalized), we have
\begin{align}
    \mathcal{L}_{BPR} &= \mathop{\mathbb{E}}_{(u, i)\sim p_{\rm pos}}-\log{\rm sigmoid}\left(s(u,i) - s(u,i^-)\right) \nonumber \\
    &= \mathop{\mathbb{E}}_{(u, i)\sim p_{\rm pos}}-\log\left(\dfrac{e^{\tilde{f(u)}^T\tilde{f(i)}}}{e^{\tilde{f(u)}^T\tilde{f(i)}} + e^{\tilde{f(u)}^T\tilde{f(i^-)}}}\right) \nonumber \\
    &=\mathop{\mathbb{E}}_{(u, i)\sim p_{\rm pos}}-\tilde{f(u)}^T\tilde{f(i)} + \log\left(e^{\tilde{f(u)}^T\tilde{f(i)}} + e^{\tilde{f(u)}^T\tilde{f(i^-)}}\right) \nonumber \\
    &\geq\mathop{\mathbb{E}}_{(u, i)\sim p_{\rm pos}}\left[-1 + \log\left(e^1 + e^{\tilde{f(u)}^T\tilde{f(i^-)}}\right)\right]\\
    &\geq -1 + \int_{\mathcal{S}^{d-1}}\int_{\mathcal{S}^{d-1}}\log\left(e + e^{x^Ty}\right){\rm d}\sigma_{d-1}(x){\rm d}\sigma_{d-1}(y).
\end{align}
According to the definition of perfect alignment, the equality in Equation~(4) is satisfied if and only if $f$ is perfectly aligned.
According to Lemma 2 in~\cite{wang2020understanding}, Equation~(5) is satisfied if and only if the feature distribution induced by $f$ is $\sigma_{d-1}$ ($f$ is perfectly uniform).
Therefore, $\mathcal{L}_{BPR} \geq$ a constant independent of $f$, where equality is satisfied if and only if $f$ is perfectly aligned and uniform.
\end{proof}

Considering the quantified metrics in Section 2.2 have been shown to be well aligned with perfect alignment and uniformity~\cite{wang2020understanding}, this theorem shows that the BPR loss indeed favors lower $l_{\rm align}$ and $l_{\rm uniform}$.
Next, we will empirically show the learning dynamics of different CF methods in terms of alignment and uniformity.

\begin{figure}[t!]
\centering
\includegraphics[trim={0 0 0 0}, clip, width=1\columnwidth]{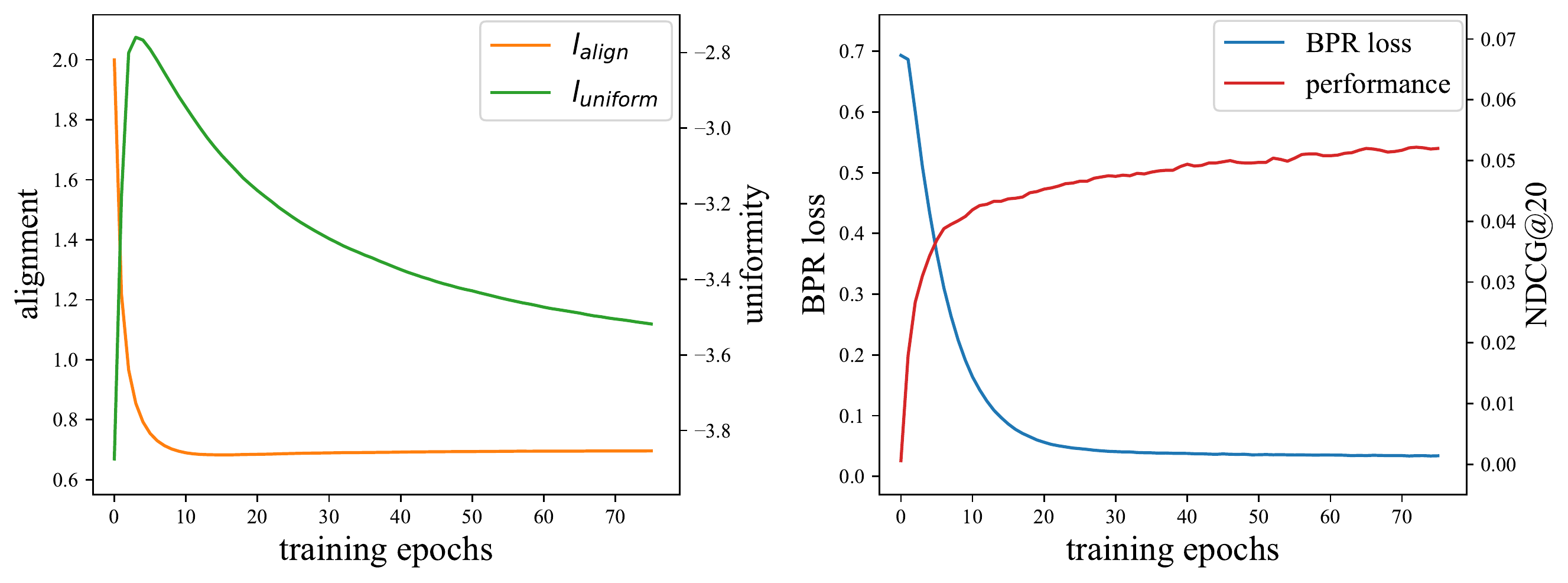}
\caption{The trends of $l_{\rm align}$ and $l_{\rm uniform}$ during training (left) and the learning curve (right) when optimizing the BPR loss on the Beauty dataset.}
\label{fig:bpr}
\vspace{-0mm}
\end{figure}

\subsection{Empirical Observations}
We use the BPR loss to train a matrix factorization (MF) model on the Beauty dataset\footnote{More information about the dataset and metrics will be detailed in Section 5.1.}.
The encoder here is a simple embedding table that maps IDs to embeddings.
Figure~\ref{fig:bpr} shows how these two properties\footnote{Calculations of alignment and uniformity losses in CF will be detailed in Appendix.}, the BPR loss, and the recommendation performance (NDCG@20), change during training.
First, we find the randomly initialized encoder is poorly aligned but well uniform (the initial uniformity loss is low).
With the optimization of the BPR loss, the alignment loss decreases quickly and results in the increase of the uniformity loss.
As the alignment loss becomes stable, the uniformity loss begins to decrease.
Overall, the recommendation performance improves as better alignment and uniformity are achieved.
This empirically validates the analyses in Section 3.1 that the BPR loss indeed optimizes for lower $l_{\rm align}$ and $l_{\rm uniform}$.

Besides the simplest MF encoder with the BPR loss (BPRMF), different CF methods may have distinct learning trajectories.
We further visualize the alignment and uniformity metrics every epoch\footnote{The start point of each method is the status after 5 epochs of training. We do not draw the first few points because they are far away from the main area.} for 4 typical CF methods on Beauty, as shown in Figure~\ref{fig:observation}.
BPRMF denotes the simplest MF encoder with the BPR loss as mentioned above.
BPR-DS~\cite{rendle2014improving} enhances BPRMF by adopting a dynamic negative sampling strategy that makes the sampling probability proportional to the predicted score.
LGCN~\cite{he2020lightgcn} utilizes graph neural network (GNN) as the encoder and uses the standard BPR training strategy.
ENMF~\cite{chen2020efficient} leverages all the negative interactions and devises an efficient approach to optimize the mean squared error (MSE) loss.
The stars in Figure~\ref{fig:observation} indicate the converged points of different models, and we annotate NDCG@20 in parentheses.
We mainly have the following observations:
\begin{itemize}
    \item The optimization of BPR focuses more on uniformity (discriminating between positive and negative interactions) but does not continuously pushes positive user-item pairs closer.
    \item BPR-DS samples more difficult negative items, and hence leads to lower uniformity loss and better performance. But hard negatives also make it difficult to align positive user-item pairs (higher alignment loss).
    \item LGCN aggregates the neighborhood information and hence achieves superior alignment even in the beginning. This explains why LGCN generally performs well with the BPR loss. The GNN encoder structure is good at alignment, while the BPR loss does well on uniformity. Although the training procedure hurts alignment and the final uniformity is worse than BPRMF, the ending alignment is still remarkable, which leads to better performance accordingly.
    \item Different from the above pairwise methods, ENMF directly optimizes MSE and leverages all the negative interactions, which pushes the scores of positive user-item pairs to 1 but not just greater than negative pairs like BPR. This whole-data based training benefits the optimization of alignment to a large extent while maintaining promising uniformity, and hence yields superior performance. But such pointwise optimization also hurts uniformity at the later training stage.
\end{itemize}

\begin{figure}[t!]
\centering
\includegraphics[trim={0 0 0 0}, clip, width=0.95\columnwidth]{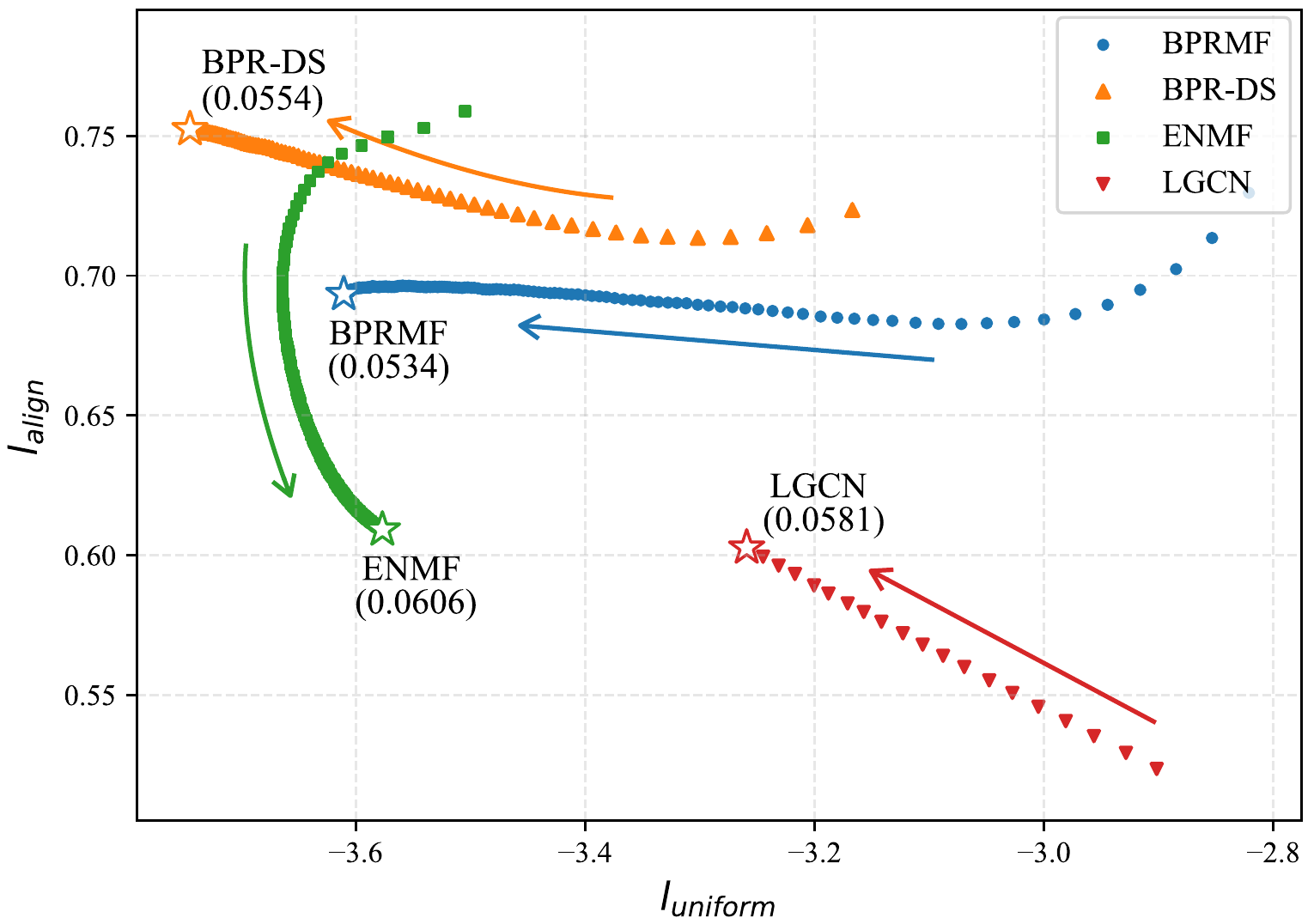}
\caption{$l_{\rm align}$-$l_{\rm uniform}$ plot for different CF methods during training. We visualize these two metrics every epoch, and the stars indicate the converged points. We also annotate NDCG@20 for each model in parentheses (higher numbers are better). For $l_{\rm align}$ and $l_{\rm uniform}$, lower numbers are better.}
\label{fig:observation}
\end{figure}

According to the above observations, we find different CF methods have distinct learning dynamics in terms of alignment and uniformity.
Compared to the standard BPRMF, BPR-DS is better at uniformity but leads to worse alignment; LGCN is better at alignment but yields worse uniformity, while both BPR-DS and LGCN achieve higher recommendation performance than BPRMF.
ENMF further gets the best performance with both promising alignment and uniformity.
This shows that user and item representations in CF indeed favor these two properties.
Achieving better alignment or uniformity both contribute to higher recommendation performance, and it can be beneficial to optimize them simultaneously.

%% file: sections/S4-Methodology.tex
\section{Directly Optimizing Alignment and Uniformity (DirectAU)}

\begin{figure}[t!]
\centering
\includegraphics[trim={0 0 0 0}, clip, width=1\columnwidth]{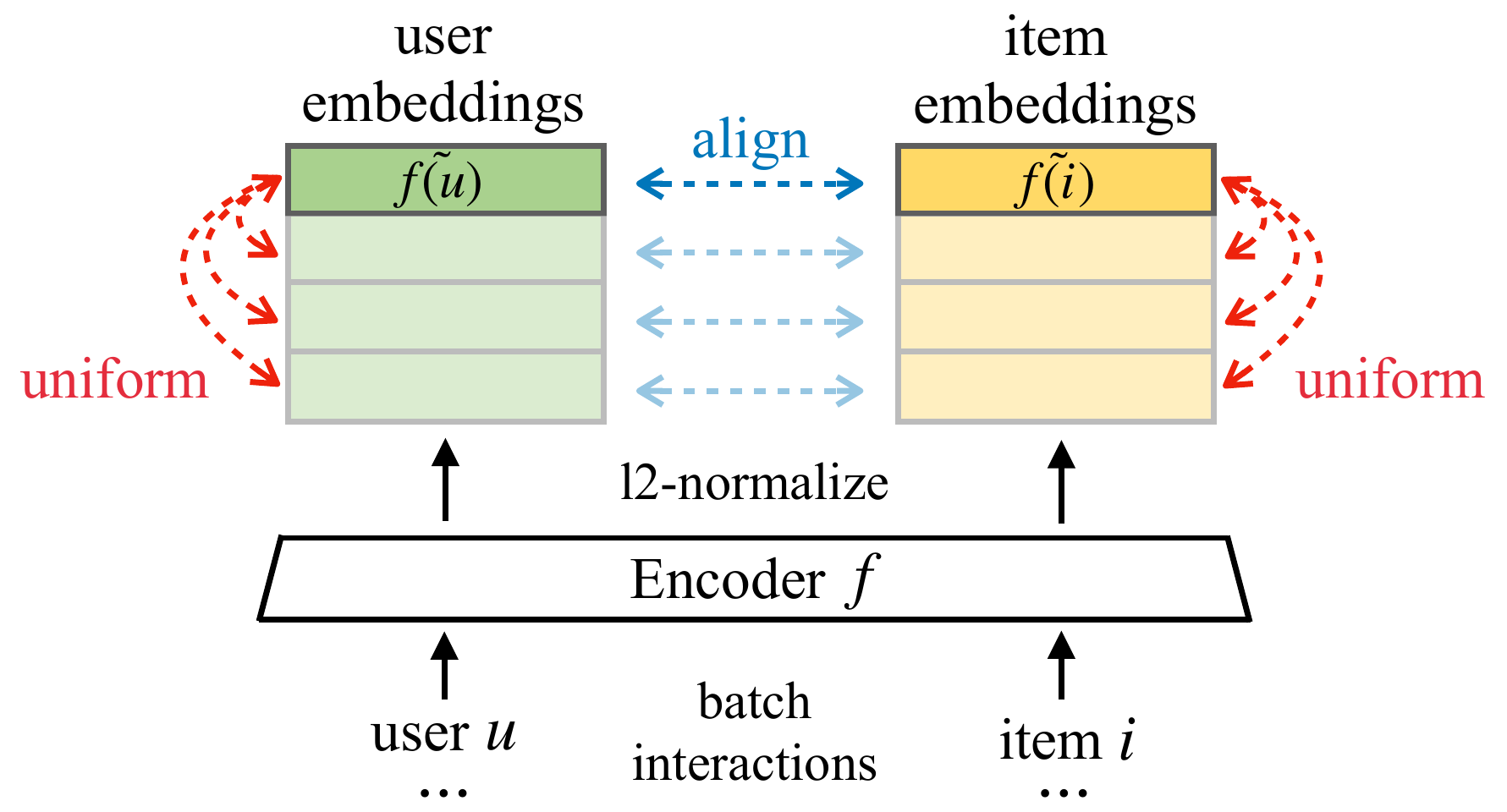}
\caption{Overview of the proposed DirectAU. We directly optimize 1) representation alignment for positive user-item pairs and 2) in-batch uniformity for users/items.}
\label{fig:model}
\vspace{-0mm}
\end{figure}

The above analyses demonstrate that both alignment and uniformity are essential to learn informative user and item representations.
This motivates us to design a new learning objective that directly optimizes these two properties to achieve better recommendation performance, named DirectAU.

Figure 1 illustrates the overall structure of the proposed framework.
The input positive user-item pairs are first encoded to embeddings and l2-normalized to the hypersphere.
We use a simple embedding table (mapping user/item IDs to embeddings) as the default encoder\footnote{Combinations with other encoders like graph neural networks will be tested in Section 5.4, and we find a simple embedding table yields remarkable performance.}.
Then, we quantify alignment and uniformity in CF as follows:
\begin{equation}
    \begin{split}
        l_{\rm align} =& \mathop{\mathbb{E}}_{(u, i)\sim p_{\rm pos}}||\tilde{f(u)} - \tilde{f(i)}||^2\\
        l_{\rm uniform} =&~\log\mathop{\mathbb{E}}_{u,u'\sim p_{\rm user}}e^{-2||\tilde{f(u)} - \tilde{f(u')}||^2} /~2~+ \\
        &~\log\mathop{\mathbb{E}}_{i,i'\sim p_{\rm item}}e^{-2||\tilde{f(i)} - \tilde{f(i')}||^2} /~2.
    \end{split}
    \label{eq:def}
\end{equation}
The alignment loss pushes up the similarity between representations of positive-related user-item pairs, while the uniformity loss measures how well the representations scatter on the hypersphere.
We separately calculate the uniformity within user representations and item representations because the data distribution of user and item might be diverse, which is more suitable to be measured respectively.
Finally, we jointly optimize these two objectives with a trade-off hyperparameter $\gamma$:
\begin{equation}
    \mathcal{L}_{\rm DirectAU} = l_{\rm align} + \gamma l_{\rm uniform}.
\end{equation}
The weight $\gamma$ controls the desired degree of uniformity, which is dependent on the characteristic of each dataset.
The learning algorithm of DirectAU can be found in Appendix.

Note that previous CF methods usually rely on negative sampling to discriminate between positive and negative interactions, while DirectAU does not need additional negative samples and only uses the input batch of positive user-item pairs.
The uniformity loss is calculated based on the in-batch pairwise distances between representations.
Using in-batch instances makes it more consistent with the actual data distribution of users and items (i.e., $p_{\rm user}, p_{\rm item}$), which has been shown to help reduce exposure bias in recommender systems~\cite{zhou2021contrastive}.
Compared to existing CF methods, DirectAU is easy to implement in the absence of negative samples, and there is only one hyper-parameter to tune (no need to consider the number of negative samples the sampling strategy).
This makes DirectAU easy to work with various application scenarios.
As for the score function, we use the dot product between user and item representations to calculate ranking scores and make recommendations, which is common in the literature~\cite{wang2020make, he2017neural, he2020lightgcn}.

%% file: sections/S5-Experiments.tex
\section{Experiments}
In this section, we conduct extensive experiments on three public datasets to validate the effectiveness of DirectAU.
We first describe the experimental settings (Section 5.1) and compare the overall top-\textit{K} recommendation performance of DirectAU with other state-of-the-art CF methods (Section 5.2).
Then, we show the learning curves when only optimizing alignment or uniformity to verify the importance of both properties (Section 5.3).
We also investigate the performance of DirectAU when integrated with other CF encoders (Section 5.4).
Finally, we provide the efficiency analyses (Section 5.5) and parameter sensitivity of DirectAU (Section 5.6).

\subsection{Experimental Settings}
\subsubsection{\textbf{Datasets}}
We use three public datasets in real-world scenarios.
All the datasets are publicly available and widely adopted in previous studies~\cite{wang2019neural, he2020lightgcn, wang2020toward, wang2022sequential}.
\begin{itemize}
    \item \textbf{Beauty}\footnote{https://jmcauley.ucsd.edu/data/amazon/links.html}: This is one of the series of product review datasets crawled from Amazon. The data is split into separate datasets by the top-level product category.
    \item \textbf{Gowalla}\footnote{http://snap.stanford.edu/data/loc-gowalla.html}: This is a check-in dataset~\cite{liang2016modeling} obtained from Gowalla, where users share their locations by checking-in.
    \item \textbf{Yelp2018}\footnote{https://www.yelp.com/dataset}: This is a business recommendation dataset, including restaurants, bars and so on. We use the transaction records after \textit{Jan. 1st, 2018} following previous work~\cite{wang2019neural, he2020lightgcn}.
\end{itemize}
For preprocessing the datasets, we remove repeated interactions and ensure each user and item to have at least 5 associated interactions.
This strategy is also widely adopted in previous work~\cite{lee2021bootstrapping, wang2019modeling}.
The statistics of datasets after preprocessing are summarized in Table \ref{tab:dataset}.

\subsubsection{\textbf{Baselines}}
We compare the performance of DirectAU with various state-of-the-art CF methods:
\begin{itemize}
    \item \textbf{BPRMF}~\cite{rendle2009bpr}: This is a typical negative-sampling method that optimizes MF with a pairwise ranking loss, where the negative item is randomly sampled from the item set.
    \item \textbf{BPR-DS}~\cite{rendle2014improving}: This method enhances BPRMF by adopting the dynamic sampling strategy, where negative items with higher prediction scores are more likely to be sampled.
    \item \textbf{ENMF}~\cite{chen2020efficient}: This is a MF-based model that uses all the unobserved interactions as negative samples without negative sampling. An efficient learning algorithm that minimizes the MSE loss is introduced to learn from the whole data.
    \item \textbf{RecVAE}~\cite{shenbin2020recvae}: This method is based on the variational autoencoder that reconstructs partially-observed user vectors, which introduces several techniques to improve M-VAE~\cite{liang2018variational}.
    \item \textbf{LGCN}~\cite{he2020lightgcn}: This is a simplified graph convolution network for CF that performs linear propagation between neighbors on the user-item bipartite graph.
    \item \textbf{DGCF}~\cite{wang2020disentangled}: This is a state-of-the-art GNN-based method that introduces disentanglement on top of LGCN, which models the intent-aware interaction graphs and encourages independence of different intents.
    \item \textbf{BUIR}~\cite{lee2021bootstrapping}: This is a state-of-the-art negative-sample-free CF method that learns user and item embeddings with only positive interactions.
    \item \textbf{CLRec}~\cite{zhou2021contrastive}: This is a recently proposed method based on contrastive learning, which adopts the InfoNCE loss to address the exposure bias in recommender systems.
\end{itemize}

\begin{table}
\normalsize
\tabcolsep=5.5pt
  \centering
  \caption{Statistics of datasets.}
  \label{tab:dataset}
  \begin{tabular}{lccccc}
    \toprule
    Dataset & \tabincell{c}{\#user\\($|\mathcal{U}|$)} & \tabincell{c}{\#item\\($|\mathcal{I}|$)} & \tabincell{c}{\#inter.\\($|\mathcal{R}|$)} & \tabincell{c}{avg. inter.\\per user} & density\\
    \midrule
    Beauty & 22.4k & 12.1k & 198.5k & 8.9 & 0.07\%\\
    Gowalla & 29.9k & 41.0k & 1027.4k & 34.4 & 0.08\%\\
    Yelp2018 & 31.7k & 38.0k & 1561.4k & 49.3 & 0.13\%\\
  \bottomrule
\end{tabular}
\vspace{-0mm}
\end{table}

\subsubsection{\textbf{Evaluation Protocols}}
Following the common practice~\cite{wang2019neural, he2020lightgcn, he2017neural}, for each dataset, we randomly split each user's interactions into training/validation/test sets with the ratio of 80\%/10\%/10\%.
To evaluate the performance of top-\textit{K} recommendation, we employ Recall and Normalized Discounted Cumulative Gain (NDCG) as evaluation metrics.
Recall@\textit{K} measures how many target items are retrieved in the recommendation result, while NDCG@\textit{K} further concerns about their positions in the ranking list.
Note that we consider the ranking list of all items (except for the training items in the user history) instead of ranking a smaller set of random items together with the target items, as suggested by recent work~\cite{DBLP:conf/kdd/KricheneR20}.
We repeat each experiment 5 times with different random seeds and report the average score.

\subsubsection{\textbf{Implementation Details}}
We use the RecBole~\cite{recbole} framework to implement all the methods for fair comparisons.
Adam is used as the default optimizer and the maximum number of epochs is set to 300.
Early stop is adopted if NDCG@20 on the validation dataset continues to drop for 10 epochs.
We set the embedding size to $64$ and the learning rate to 1e$^{-3}$ for all the methods.
The training batch size is set to 256 on Beauty and 1024 on the other two datasets.
The weight decay is tuned among [0, 1e$^{-8}$, 1e$^{-6}$, 1e$^{-4}$].
The default encoder $f$ in DirectAU is a simple embedding table that maps user/item IDs to embeddings.
The weight $\gamma$ of $l_{\rm uniform}$ in DirectAU is tuned within [0.2, 0.5, 1, 2, 5, 10].
As for baseline-specific hyper-parameters, we tune them in the ranges suggested by the original paper.
All the parameters are initialized by xavier initialization.
Codes are publicly available\footnote{https://github.com/THUwangcy/DirectAU}.

\begin{table*}[t!]
\tabcolsep=6.5pt
\centering
\caption{Top-\textit{K} recommendation performance on three datasets. The best results are in bold face, and the best baselines are underlined. The superscripts $^{**}$ indicate $p\leq0.01$ for the paired t-test of DirectAU vs. the best baseline (the relative improvements are denoted as Improv.).}
\begin{tabular}{C{1cm}L{1.6cm}cccccccclc}
    \toprule
    \multicolumn{2}{c}{Setting} & \multicolumn{8}{c}{Baseline Methods} & \multicolumn{2}{c}{Ours}\cr 
    \cmidrule(lr){1-2} \cmidrule(lr){3-10} \cmidrule(lr){11-12}
    Dataset & \multicolumn{1}{c}{Metric} & BPRMF & BPR-DS & ENMF & RecVAE & LGCN & DGCF & BUIR & CLRec & DirectAU & Improv. \cr
    \midrule
    \multirow{6}{*}{\rotatebox{90}{Beauty}} 
    & Recall@10 & 0.0806 & 0.0816 & 0.0915 & 0.0824 & 0.0863 & 0.0897 & 0.0816 & \underline{0.0937} & \textbf{0.1002}$^{**}$ & 6.94\% \cr
    & Recall@20 & 0.1153 & 0.1181 & 0.1282 & 0.1145 & 0.1201 & 0.1283 & 0.1204 & \underline{0.1337} & \textbf{0.1400}$^{**}$ & 4.74\% \cr
    & Recall@50 & 0.1763 & 0.1745 & 0.1914 & 0.1712 & 0.1819 & 0.1958 & 0.1866 & \underline{0.1996} & \textbf{0.2062}$^{**}$ & 3.33\% \cr
    \cmidrule(lr){2-2}\cmidrule(lr){3-10}\cmidrule(lr){11-12}
    & NDCG@10 & 0.0444 & 0.0459 & 0.0511 & 0.0486 & 0.0484 & 0.0501 & 0.0457 & \underline{0.0547} & \textbf{0.0582}$^{**}$ & 6.44\% \cr
    & NDCG@20 & 0.0534 & 0.0554 & 0.0606 & 0.0570 & 0.0581 & 0.0600 & 0.0556 & \underline{0.0651} & \textbf{0.0686}$^{**}$ & 5.38\% \cr
    & NDCG@50 & 0.0658 & 0.0670 & 0.0736 & 0.0686 & 0.0699 & 0.0738 & 0.0692 & \underline{0.0786} & \textbf{0.0820}$^{**}$ & 4.33\% \cr
    \midrule
    \multirow{6}{*}{\rotatebox{90}{Gowalla}} 
    & Recall@10 & 0.0866 & 0.1132 & 0.1149 & 0.1211 & 0.1289 & \underline{0.1301} & 0.0798 & 0.1215 & \textbf{0.1394}$^{**}$ & 7.15\% \cr
    & Recall@20 & 0.1263 & 0.1637 & 0.1671 & 0.1771 & 0.1871 & \underline{0.1889} & 0.1164 & 0.1755 & \textbf{0.2014}$^{**}$ & 6.63\% \cr
    & Recall@50 & 0.2040 & 0.2593 & 0.2675 & 0.2768 & \underline{0.2934} & 0.2919 & 0.1917 & 0.2813 & \textbf{0.3127}$^{**}$ & 6.56\% \cr
    \cmidrule(lr){2-2}\cmidrule(lr){3-10}\cmidrule(lr){11-12}
    & NDCG@10 & 0.0622 & 0.0814 & 0.0797 & 0.0845 & 0.0930 & \underline{0.0939} & 0.0570 & 0.0868 & \textbf{0.0991}$^{**}$ & 5.56\% \cr
    & NDCG@20 & 0.0736 & 0.0961 & 0.0953 & 0.1007 & 0.1097 & \underline{0.1099} & 0.0676 & 0.1022 & \textbf{0.1170}$^{**}$ & 6.44\% \cr
    & NDCG@50 & 0.0926 & 0.1196 & 0.1200 & 0.1251 & 0.1356 & \underline{0.1358} & 0.0858 & 0.1281 & \textbf{0.1442}$^{**}$ & 6.20\% \cr
    \midrule
    \multirow{6}{*}{\rotatebox{90}{Yelp2018}} 
    & Recall@10 & 0.0416 & 0.0533 & \underline{0.0596} & 0.0495 & 0.0508 & 0.0519 & 0.0444 & 0.0547 & \textbf{0.0684}$^{**}$ & 14.83\% \cr
    & Recall@20 & 0.0693 & 0.0864 & \underline{0.0957} & 0.0820 & 0.0833 & 0.0849 & 0.0737 & 0.0890 & \textbf{0.1096}$^{**}$ & 14.55\% \cr
    & Recall@50 & 0.1293 & 0.1572 & \underline{0.1710} & 0.1494 & 0.1534 & 0.1575 & 0.1386 & 0.1606 & \textbf{0.1935}$^{**}$ & 13.16\% \cr
    \cmidrule(lr){2-2}\cmidrule(lr){3-10}\cmidrule(lr){11-12}
    & NDCG@10 & 0.0335 & 0.0423 & \underline{0.0482} & 0.0395 & 0.0406 & 0.0409 & 0.0349 & 0.0436 & \textbf{0.0553}$^{**}$ & 14.77\% \cr
    & NDCG@20 & 0.0428 & 0.0534 & \underline{0.0603} & 0.0504 & 0.0514 & 0.0521 & 0.0448 & 0.0551 & \textbf{0.0691}$^{**}$ & 14.53\% \cr
    & NDCG@50 & 0.0602 & 0.0740 & \underline{0.0821} & 0.0698 & 0.0717 & 0.0732 & 0.0636 & 0.0758 & \textbf{0.0933}$^{**}$ & 13.67\% \cr
    \bottomrule
\end{tabular}
\label{tab:exp}
\end{table*}

\subsection{Overall Performance}
Table~\ref{tab:exp} shows the performance of different baseline CF methods and our DirectAU.
From the experimental results, we mainly have the following observations.

Firstly, it is surprising that directly optimizing alignment and uniformity yields such impressive performance improvements, given that most baselines come from studies in recent two years.
This demonstrates that these two properties strongly agree with the representation quality in CF, and current models might not address both alignment and uniformity well, which leads to inferior results.
Compared to state-of-the-art CF methods, DirectAU is not only conceptually simple but also empirically effective.

Secondly, we find the best baseline varies in different datasets.
The contrastive learning based CLRec is effective on Beauty; while the GNN-based DGCF takes advantage on Gowalla; and ENMF achieves remarkable performance on the largest dataset Yelp2018.
This shows that the characteristics of different CF models may suit different application scenarios.
On the contrary, DirectAU is capable of directly adjusting the balance between alignment and uniformity, leading to consistently the best performance on all three datasets.

Thirdly, comparing different kinds of baselines, methods with more complex encoders do not always benefit the performance.
The most complex model DGCF is only the most effective on Gowalla but generally costs much more time for training.
Differently, methods focusing on the learning objective (e.g., ENMF, CLRec) are more robust and usually yields promising results.
This shows the importance of designing suitable loss functions rather than sophisticated encoders.
The effectiveness of DirectAU also suggests that it is useful to understand the desired properties of representations in CF, which benefit the design of more powerful loss functions.

\begin{figure}[t!]
\centering
\includegraphics[trim={0 0.cm 0 0}, clip, width=1\columnwidth]{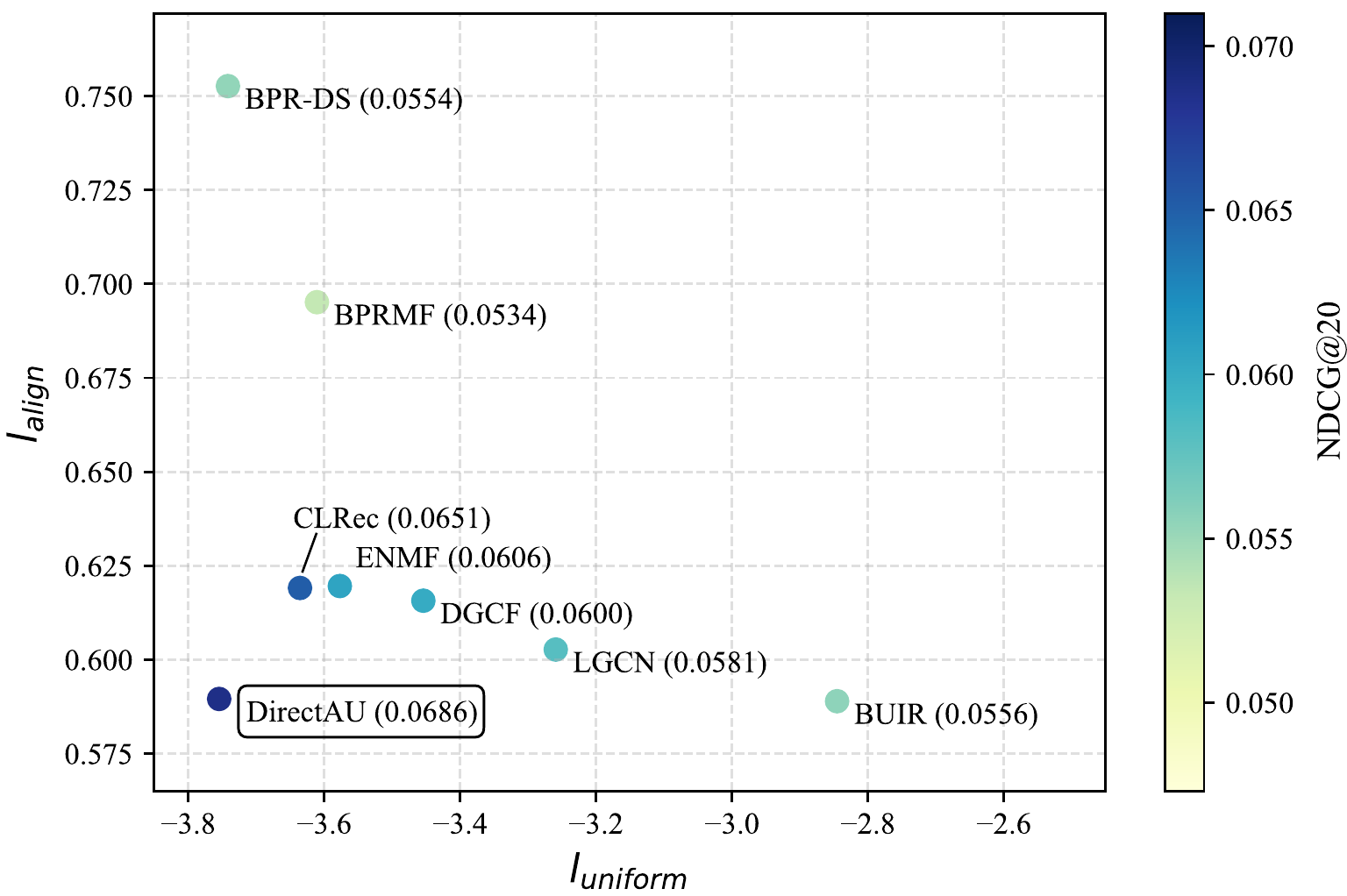}
\caption{$l_{\rm align}$-$l_{\rm uniform}$ plot of different CF models on Beauty. For both $l_{\rm align}$ and $l_{\rm uniform}$, lower numbers are better. Colors and numbers in parentheses indicate NDCG@20.}
\label{fig:overall}
\end{figure}

\begin{figure*}[ht]
\centering
\includegraphics[trim={0 0.2cm 0 0}, clip, width=0.99\textwidth]{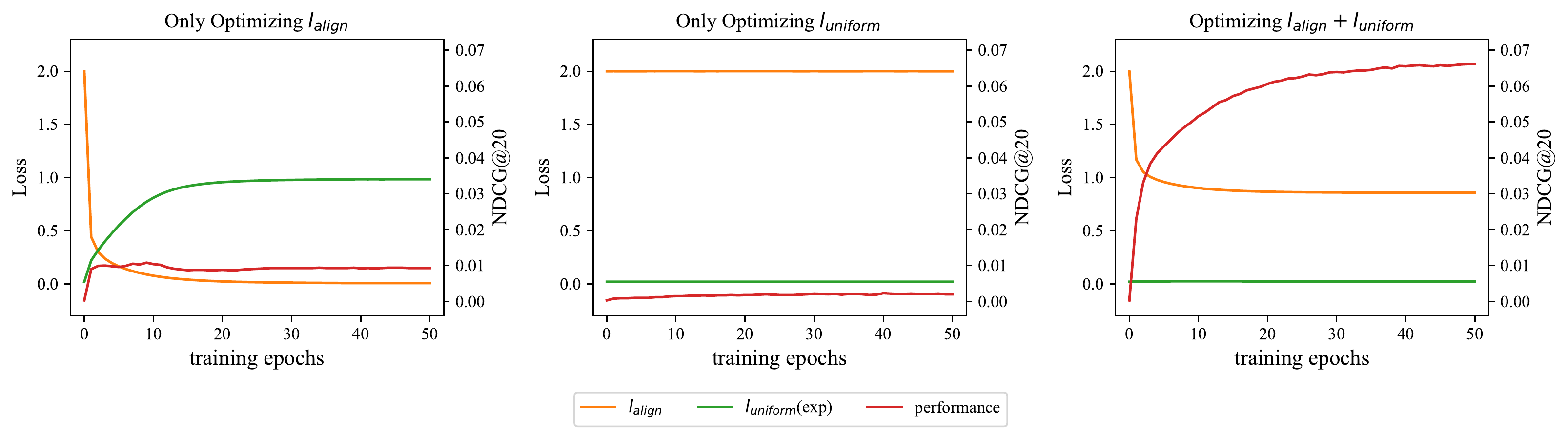}
\caption{Learning curves when only optimizing the alignment loss (left), only optimizing the uniformity loss (middle), and optimizing both of the losses (right) on Yelp2018. $l_{\rm uniform}$ is exponentiated for better visualization. The encoder yields poor performance when only one of alignment and uniformity is optimized. Both of the properties are important to learn high-quality user and item representations.}
\label{fig:ablation}
\end{figure*}

Furthermore, in Figure~\ref{fig:overall}, we show the alignment and uniformity of different CF methods\footnote{RecVAE is not included because it is a generative method without item embeddings. The alignment and uniformity metrics are invalid under our definition.} along with their recommendation performance on Beauty.
Overall, we can see methods with both better alignment and uniformity achieve better performance.
ENMF and CLRec become two strong baselines because of the balance between these two properties.
DGCF mainly improves uniformity on top of LGCN by introducing disentanglement to the representations.
The recent novel method BUIR achieves promising results without negative samples mainly due to the superiority in alignment.
But without the supervision signals from negative samples, the uniformity of BUIR is poor.
Compared to state-of-the-art CF methods, DirectAU achieves the lowest alignment and uniformity losses and yields the best performance.
This verifies the causal effect of alignment and uniformity on the representation quality in CF.

\begin{table}[t]
  \normalsize
  \tabcolsep=4pt
  \centering
  \caption{Performance comparison of different encoders when integrated with the proposed DirectAU loss.}
  \label{tab:integration}
  \begin{tabular}{lcccc}
    \toprule
    \multirow{2}{*}{Method} & \multicolumn{2}{c}{Beauty} & \multicolumn{2}{c}{Gowalla} \\
    \cmidrule(lr){2-3}\cmidrule(lr){4-5}
    & Recall@20 & NDCG@20 & Recall@20 & NDCG@20 \\
    \midrule
    BPRMF & 0.1153 & 0.0534 & 0.1263 & 0.0736 \\
    +DirectAU & \textbf{0.1400}$^{**}$ & \textbf{0.0686}$^{**}$ & \textbf{0.2014}$^{**}$ & \textbf{0.1170}$^{**}$ \\
    \midrule
    LGCN-1 & 0.1211 & 0.0560 & 0.1769 & 0.1033 \\
    +DirectAU & \textbf{0.1444}$^{**}$ & \textbf{0.0700}$^{**}$ & \textbf{0.2036}$^{**}$ & \textbf{0.1184}$^{**}$ \\
    \midrule
    LGCN-2 & 0.1201 & 0.0581 & 0.1871 & 0.1097 \\
    +DirectAU & \textbf{0.1455}$^{**}$ & \textbf{0.0707}$^{**}$ & \textbf{0.2043}$^{**}$ & \textbf{0.1191}$^{**}$ \\
    \bottomrule
  \end{tabular}
\end{table}

\subsection{Importance of Both Alignment and Uniformity Losses}
To show that both properties are important to learn informative encoders, Figure~\ref{fig:ablation} gives the learning curves when 1) only optimizing the alignment loss, 2) only optimizing the uniformity loss, and 3) optimizing both of the losses on Yelp2018.
If only alignment is considered (left), the encoder achieves perfect alignment ($l_{\rm align}$ approaches 0) but suffers a degeneration in uniformity.
As a result, the recommendation performance only improves a little at the beginning and then converges to poor results.
If only uniformity is considered (middle), the encoder maintains uniformity (randomly initialized embeddings are well uniform) but does not improve alignment.
Hence, the performance is even worse than only optimizing $l_{\rm align}$.
Differently, when optimizing both alignment and uniformity (right), the encoder keeps promising uniformity and continuously improves alignment at the same time.
As a result, the representation quality steadily increases and boosts the recommendation performance.
These trends demonstrate the importance of addressing both alignment and uniformity in CF.

\begin{table}[t]
  \normalsize
  \tabcolsep=7.2pt
  \centering
  \caption{Efficiency comparison on Yelp2018, including the average training time per epoch, the number of epochs to converge, and the total training time (s: second, m: minute, h: hour).}
  \label{tab:efficiency}
  \begin{tabular}{lcccccc}
    \toprule
    Method & time/epoch & \#epoch & total time \\
    \midrule
    BPRMF & 29.8s & 59 & 29m \\
    ENMF & 24.8s & 89 & 36m \\
    LGCN & 228.6s & 107 & 6h48m \\
    \midrule
    DirectAU & 37.3s & 50 & 31m \\
    \bottomrule
  \end{tabular}
\end{table}

\subsection{Integration with Other CF Encoders}
In the main experiments (Table~\ref{tab:exp}), we optimize the DirectAU loss with a simple MF encoder (i.e., embedding table).
This raises the question that whether it is also beneficial to directly optimize alignment and uniformity for other CF encoders.
Here we take MF and LGCN with different numbers of layers as the interaction encoder, respectively.
Table~\ref{tab:integration} shows the performance of these methods with their original losses and corresponding variants with the DirectAU loss.
LGCN-\textit{X} means the LGCN encoder with \textit{X} GNN layers.
We can see DirectAU consistently brings remarkable improvements to each encoder.
Besides, when integrated with more powerful encoders like LGCN-2, DirectAU achieves higher performance than the default MF encoder.
This shows the generalization ability of the proposed learning framework.
Meanwhile, we find the relative improvements are the most significant for the simplest MF encoder.
In the Gowalla dataset, it is impressive that MF+DirectAU leads to 59.2\% improvements than the original MF on average; while LGCN-2+DirectAU only brings around 8.9\% improvements.
This verifies the importance of choosing proper learning objectives in CF.
With the help of the DirectAU loss, a simple MF encoder can also learn high-quality representations, and hence achieves comparable results with the complex LGCN-2 encoder.
Considering the balance of effectiveness and efficiency, we still choose MF as the default encoder in the following analyses.

\subsection{Efficiency Analyses}
Here we compare the training efficiency of DirectAU with BPRMF and other two state-of-the-art CF models, i.e., ENMF and LGCN, which are both relatively efficient in their respective categories.
In Table~\ref{tab:efficiency}, we present the average training time per epoch, the number of epochs to converge, and the total training time on the largest dataset Yelp2018.
The efficiency experiments are conducted on the same machine (Intel Core 12-core CPU of 3.5GHz and single NVIDIA GeForce GTX 1080 Ti GPU).
We compare different methods under the same implementation framework and the setting of batch size is fixed to 256 to ensure fairness.
The results show that ENMF is the most efficient in terms of the training time per epoch, which results from the specifically designed learning algorithm.
The graph-based LGCN is much slower because of the neighborhood aggregation in each iteration, even if LGCN performs linear propagation for simplicity.
Our DirectAU needs a little more training time per epoch than BPRMF and ENMF mainly due to the calculation of the uniformity loss.
However, DirectAU generally converges fast and the total time is similar with BPRMF and ENMF, which is much faster than LGCN.
Thus, DirectAU is relatively efficient for its simplicity, and we believe the performance gains justify the runtime costs in practice.

\subsection{Parameter Sensitivity}
DirectAU introduces a hyper-parameter $\gamma$ that controls the weight of the uniformity loss.
It is worth noting that this is the only hyper-parameter to tune for DirectAU, which does not rely on negative sampling like previous CF methods.
Therefore, there is no need to consider the number of negative samples and the sampling strategy.
This makes DirectAU easy to use in real-world applications.
Figure~\ref{fig:param} shows how the performance changes when varying this hyper-parameter on the three datasets.
We can observe a similar trend that the performance increases first and then decreases.
Different datasets suit different degrees of uniformity, which depend on the characteristics of datasets.
We find higher uniformity weights might be preferable for datasets with more average interactions per user (i.e., Gowalla, Yelp2018), in which case representations might be more likely to be pushed closer due to the alignment loss.
Note that the range of $\gamma$ is not restricted from 0.2 to 10, which may need wider ranges and fine-grained steps in practice.

%% file: sections/S6-RelatedWork.tex
\section{Related Work}
\subsection{Collaborative Filtering}
Collaborative filtering (CF) plays an essential role in recommender systems~\cite{schafer2007collaborative}.
The core idea of CF is that similar users tend to have similar preferences.
Different from content-based filtering methods, CF does not rely on user and item profiles to make recommendations, and hence is flexible to work in various domains.
One of the primary methods for CF is the latent factor model, which learns latent user and item representations from observed interactions.
The predicted score of an unobserved user-item pair is derived by the similarity (e.g., dot product) between the user and item representation.
Traditional methods are mainly based on matrix factorization (MF)~\cite{koren2008factorization, koren2009matrix}.
With the development of neural networks, neural CF models begin to emerge to learn more powerful user/item representations~\cite{he2017neural, zhang2019deep}.
Besides, graph neural networks attract increasing attention recently, and a number of graph-based CF models have been proposed~\cite{he2020lightgcn, wang2019neural, wang2020disentangled, wu2019session}.
The observed user-item interactions are taken as a bipartite graph, and graph neural networks help to capture high-order connection information.

\begin{figure}[t]
\centering
\includegraphics[trim={0 0.8cm 0 0}, clip, width=1.0\columnwidth]{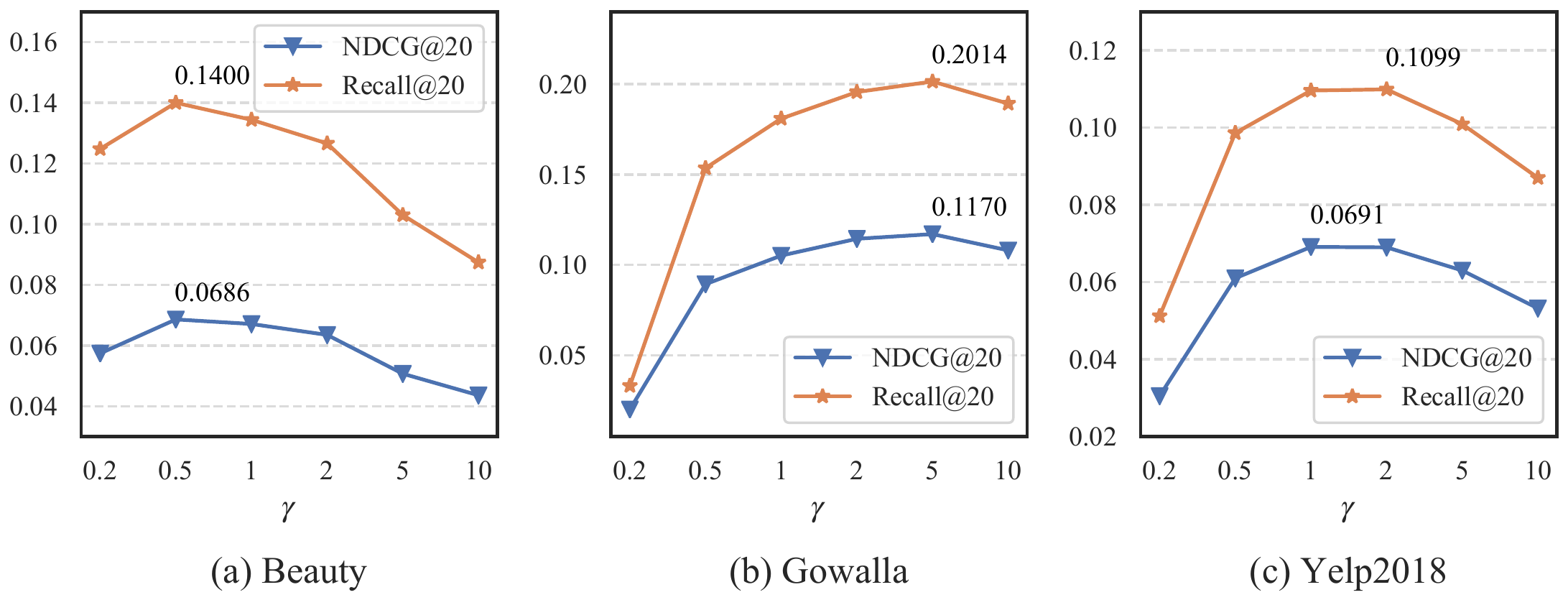}
\caption{Parameter sensitivity with regard to the weight of $l_{\rm uniform}$ in DirectAU.}
\label{fig:param}
\end{figure}

Existing studies in CF mainly focus on the model structure of the encoder but pay less attention to other components like the learning objective and the negative sampling strategy, which also contribute to the final performance.
Some recent works~\cite{chen2020efficient, lee2021bootstrapping, liu2021contrastive, mao2021simplex} begin to investigate alternative learning paradigms.
For example, ENMF~\cite{chen2020efficient} devises an efficient approach to optimize the MSE loss based on the whole data.
BUIR~\cite{lee2021bootstrapping} presents a novel asymmetric structure to learn from positive-only data.
CLRec~\cite{liu2021contrastive} adopts the InfoNCE loss in contrastive learning to address the exposure bias in recommender systems.
In this paper, we also focus on the learning objective in CF.
Differently, we are the first to investigate the desired properties of representations in CF from the perspective of alignment and uniformity.
And a new loss function that directly optimizes these two properties is proposed based on the analyses results.

\subsection{Alignment and Uniformity in Contrastive Representation Learning}
Unsupervised contrastive representation learning has witnessed great success in recent years~\cite{grill2020bootstrap}.
Studies in this literature usually aim to learn informative representations on the unit hypersphere based on self-supervised tasks.
Recent work~\cite{wang2020understanding} identifies two key properties related to the quality of representations, namely alignment and uniformity.
Similar instances are expected to have similar representations (alignment), and the distribution of representations is preferred to preserve as much information as possible (uniformity).
Alignment is usually easy to be achieved (e.g., mapping all the inputs to the same representations), but it is hard to maintain uniformity at the same time.
Previous representation learning strategies can be seen to preserve uniformity in different ways, such as discriminating from negative samples~\cite{gao2021simcse} and feature decorrelation~\cite{zbontar2021barlow}.
Directly matching uniformly sampled points on the unit hypersphere is also shown to provide good representations~\cite{bojanowski2017unsupervised}.
However, to the best of our knowledge, there still lacks thorough investigations towards alignment and uniformity in CF.
This work theoretically shows the connection between the typical BPR loss and these two properties.
Besides, our analyses towards learning dynamics of different CF methods help understand the rationales of existing CF methods and design new learning objectives.

%% file: sections/S7-Conclusion.tex
\section{Conclusion}
In this paper, we investigate the desired properties of representations in collaborative filtering (CF).
Specifically, we propose to measure the representation quality in CF from the perspective of alignment and uniformity, inspired by recent progress in contrastive representation learning.
We first theoretically reveal the connection between the commonly adopted BPR loss and these two properties.
Then, we empirically analyze the learning dynamics of typical CF methods in terms of alignment and uniformity.
We find different methods may be good at different aspects, while either better alignment or better uniformity leads to higher recommendation performance.
Based on the analyses results, a loss function that directly optimizes these two properties is proposed and experimented to be effective.
A simple matrix factorization model with the proposed loss function achieves superior performance compared to state-of-the-art CF methods.
We hope this work could inspire the CF community to pay more attention to the learning paradigm via in-depth analyses towards the representation quality.

In the future, we will investigate other learning objectives that also favor alignment and uniformity to further improve effectiveness and efficiency.

%% file: sections/Z-Appendix.tex
\section{Appendix}
In the appendix, we first show the learning algorithm of the proposed DirectAU.
Then, we detail the calculation of the alignment and uniformity losses when measuring the entire learned embeddings in CF.

\subsection{Learning Algorithm of DirectAU}
Algorithm~\ref{alg} shows the learning algorithm of DirectAU.
PyTorch-style pseudocodes to calculate alignment and uniformity losses during training are also given to facilitate reproducibility.

\begin{algorithm}[h]
  \caption{Learning algorithm of DirectAU (PyTorch style)}  
  \label{alg}  
  \begin{algorithmic}[1]  
    \Require  
      user-item interactions data $\mathcal{R}$; 
      structure of encoder network $f$;
      weight of the uniformity loss $\gamma$;
      embedding dimension $d$.
    \Ensure  
      encoder parameters $\theta$  
    \State Randomly initialize all parameters.
    \For{each mini-batch with $n$ user-item pairs $(u,i)\in\mathcal{R}$}
        \State Get user and item embeddings $f(u), f(i)$
        \State $\mathbf{x}=f(u)~/~||f(u)||,~~~\mathbf{y}=f(i)~/~||f(i)||$
        \State $\mathcal{L}_{\rm DirectAU}=\textproc{Align}(\mathbf{x}, \mathbf{y}) + \gamma\cdot(\textproc{Uni}(\mathbf{x}) + \textproc{Uni}(\mathbf{y}))~/~2$
        \State Update the encoder $f$ by gradient descent
    \EndFor \\
    
    \Function{align}{\texttt{\footnotesize x, y}} \hfill \textcolor[RGB]{75,123,128}{\# alignment loss}
        \State \Return \texttt{\footnotesize (x - y).norm(dim=1).pow(2).mean()}
    \EndFunction 
    \Function{uni}{\texttt{\footnotesize x}} \hfill \textcolor[RGB]{75,123,128}{\# uniformity loss}
        \State \texttt{\footnotesize dist = torch.pdist(x, p=2).pow(2)}
        \State \Return \texttt{\footnotesize dist.mul(-2).exp().mean().log()}
    \EndFunction
  \end{algorithmic} 
\end{algorithm} 

\subsection{Alignment and Uniformity Calculation}
According to our definitions for alignment and uniformity in CF, i.e., Eq.(\ref{eq:def}), user-item pairs to calculate the alignment loss should sample from the distribution of positive interactions $p_{\rm pos}$, and user-user (item-item) pairs to calculate the uniformity loss should sample from the corresponding user/item distribution $p_{\rm user}/p_{\rm item}$.
Given the learned embeddings of all the users and items, the alignment loss can be directly calculated as follows:
\begin{equation}
    l_{\rm align} = \dfrac{1}{|\mathcal{R}|}\sum_{(u,i)\in\mathcal{R}}||\tilde{f(u)} - \tilde{f(i)}||^2,
\end{equation}
where $\mathcal{R}$ is the set of observed user-item interactions as mentioned in Section 2.1.
We only need to traverse all the $(u,i)$ pairs in $\mathcal{R}$, and the time complexity is $O(|\mathcal{R}|)$.

As for the calculation of uniformity, a naive and intuitive method is to sample $u,u'\in\mathcal{U}$ and $i,i'\in\mathcal{I}$.
However, this is not consistent with the definition that $u,u'\sim p_{\rm user}$ and $i,i'\sim p_{\rm item}$.
Notice that the calculation of the uniformity loss during training follows the actual $p_{\rm user}$ and $p_{\rm item}$ because the training batch is constructed based on positive interactions.
When measure the overall uniformity of the learned embeddings, we should also sample two interactions from $\mathcal{R}$ and retain the user/item side as the input pair, which ensures that $u,u'$ and $i,i'$ are sampled from corresponding distribution:
\begin{equation}
\begin{split}
    l_{\rm uniform} =&~\left(\log\dfrac{1}{|\mathcal{R}|\left(|\mathcal{R}| - 1\right)}\sum_{(u,i),(u',i')\in\mathcal{R}}e^{-2||\tilde{f(u)} - \tilde{f(u')}||^2}\right)/~2~+ \\
        &~\left(\log\dfrac{1}{|\mathcal{R}|\left(|\mathcal{R}| - 1\right)}\sum_{(u,i),(u',i')\in\mathcal{R}}e^{-2||\tilde{f(i)} - \tilde{f(i')}||^2}\right)/~2.
\end{split}
\label{eq:uni_naive}
\end{equation}
Meanwhile, this calculation method is time-consuming and contains many redundant computations.
We need to traverse the entire interaction set $\mathcal{R}$ twice and the time complexity is $O(|\mathcal{R}|^2)$, which is usually intractable in practice.
To solve this problem, we devise a method to calculate the uniformity loss by directly sampling from the user/item set together with a popularity-weighting strategy:
\begin{equation}
\begin{split}
    l_{\rm uniform} =&~\left(\log\sum_{u,u'\in\mathcal{U}}\dfrac{p(u)p(u')}{P_U}\cdot e^{-2||\tilde{f(u)} - \tilde{f(u')}||^2}\right)/~2~+ \\
        &~\left(\log\sum_{i,i'\in\mathcal{I}}\dfrac{p(i)p(i')}{P_I}\cdot e^{-2||\tilde{f(i)} - \tilde{f(i')}||^2}\right)/~2,
\end{split}
\label{eq:uni}
\end{equation}
where $p(\cdot)$ returns the number of related interactions in $\mathcal{R}$ (i.e., popularity). $P_U=\sum_{u\in\mathcal{U}}p(u)$ and $P_I=\sum_{i\in\mathcal{I}}p(i)$ is the normalization factor, respectively.
It is easy to show that Eq.(\ref{eq:uni_naive}) and Eq.(\ref{eq:uni}) are exactly equivalent, while the latter reduces the computational cost to a large extent because the scale of $\mathcal{U}/\mathcal{I}$ is usually much smaller than $\mathcal{R}$.
In this way, we can measure both alignment and uniformity of the learned embeddings efficiently.

This popularity-weighting strategy also explicitly suggests that the uniformity loss focuses more on the distances between popular users/items as expected.
Those popular users and items are more likely to be aligned very close, and it is reasonable to encourage them to scatter on the hypersphere.